\documentclass{article}
\usepackage{bbm}
\usepackage{amsmath,amsthm}
\usepackage{amsfonts}
\usepackage{mathrsfs}
\usepackage{amssymb}

 \newtheorem{thm}{Theorem}[section]
\newtheorem{rem}[thm]{Remark}
\newtheorem{definition}[thm]{Definition}
\newtheorem{lem}[thm]{Lemma}
 \newtheorem{prop}[thm]{Proposition}
\newtheorem{cor}[thm]{Corollary}

 \newtheorem{conj}[thm]{Conjecture}
 
\newcommand{\aut}{\mathrm{Aut}}
\newcommand{\paut}{\mathrm{PAut}}
\def\f#1{{\mathbb{F}}_{#1}}

\begin{document}

\title{On Deep Holes of Elliptic Curve Codes\thanks{
 The research of Jun Zhang was supported  by the National Natural Science Foundation of China under the Grant 11971321 and the National Key Research and Development Program of China under Grants 2018YFA0704703. Daqing Wan was partially supported 
 by NSF grant CCF-1900929}}

\author{Jun Zhang\thanks{Jun Zhang is with the School of Mathematical Sciences, Capital Normal University, Beijing 100048, China. Email: junz@cnu.edu.cn},
\and
Daqing Wan
\thanks{Daqing Wan is with the Department of Mathematics, University of California, Irvine, CA 92697, USA. Email: dwan@math.uci.edu}}

\date{}
\maketitle

\begin{abstract}

We give a method to construct deep holes for elliptic curve codes. For long elliptic curve codes, we conjecture that our 
construction is complete in the sense that it gives all deep holes. Some evidence and heuristics on the completeness 
are provided via the connection with problems and results in finite geometry.

\begin{flushleft}
\end{flushleft}
\end{abstract}


\section{Introduction}
The classification of deep holes in a linear code is a fundamental difficult problem in coding theory. Deciding if a given received word is a deep hole 
is already NP-hard, even for short Reed-Solomon codes. For long Reed-Solomon codes, 
this has been studied extensively, and is better understood if one assumes the MDS conjecture or the rational normal curve conjecture in finite geometry. Algebraically, Reed-Solomon codes are just algebraic geometry codes of genus zero. From this point 
of view, it is natural to study the deep hole problem for algebraic geometry codes of higher genus $g$. The difficulty naturally 
increases as the genus $g$ grows. In fact, the minimun distance is already unknown and NP-hard to determine when genus $g=1$. 
In this paper, we give the first study of the deep hole problem for elliptic curve codes, the genus $g=1$ case. 
Our main result is an explicit construction of a class of deep holes for long elliptic curve codes. 
We conjecture that our construction already gives the complete set of all deep holes for long elliptic curve codes. 
In the final section, we provide some heuristics and evidence about this completeness conjecture via its 
connection with problems and results in finite geometry.

Let $\f{q}^n$ be the $n$-dimensional vector space over the finite field $\f{q}$ of $q$ elements with characteristic $p$. For any vector (also, called \emph{word}) $ {x}=(x_1,x_2,\cdots,x_n)\in \f{q}^n$, the \emph{Hamming weight} $\mathrm{Wt}( {x})$ of $ {x}$ is defined to be the number of non-zero coordinates, i.e.,
$\mathrm{Wt}( {x})=|\left\{i\,|\,1\leqslant i\leqslant n,\,x_i\neq 0\right\}|.$
For integers $1\leq k\leq n$,
a \emph{linear $[n,k]$ code} $C$ is a $k$-dimensional linear subspace of $\f{q}^n$. The \emph{minimum distance} $d(C)$ of $C$ is the minimum Hamming weight among all non-zero vectors in $C$, i.e.,
$d(C)=\min\{\mathrm{Wt}( {c})\,|\, {c}\in C\setminus\{ {0}\}\}.$
A linear $[n,k]$ code $C\subseteq \f{q}^n$ is called an $[n,k,d]$ linear
code if $C$ has minimum distance $d$.  For error correction purpose, an $[n,k]$ code $C$ is good if its minimun distance $d$ 
is large. Ideally, for a given $[n,k]$-code $C$, one would like its minimun distance $d$ to be as large as possible. A well-known trade-off between
the parameters of a linear $[n,k,d]$ code is the Singleton bound
which states that
$$d\leqslant n-k+1.$$
 An $[n,k,d]$ code is called a \emph{maximum distance separable}
  (MDS) code if $d=n-k+1$. The MDS codes of dimension $1$ and their duals of dimension $n-1$ are called trivial MDS codes. The trivial MDS codes can have arbitrary length $n$. For length $n\leq q$, an important class of non-trivial MDS codes are Reed-Solomon codes with evaluation set $D$ chosen to be any $n$ rational points on the affine line $\mathbb{A}^1(\f{q})$. For $n=q+1$, one has the projective Reed-Solomon code which is also an MDS code. For length $n\geq q+2$, one does not 
 expect any non-trivial MDS code for odd $q$. This is the main part of the long standing MDS conjecture proposed by Segre~\cite{Segre55}.
  \begin{conj}[MDS conjecture]
  	The length $n$ of non-trivial MDS codes over the finite field $\f{q}$ cannot exceed $q+1$ with two exceptions: for $k\in\{3,q-1\}$ and even $q$ the length can reach $q+2.$
  \end{conj}
 
 This conjecture remains open in general, although a lot of progress has been made. In particular, it is known to be true if $q$ is a 
 prime, see~\cite{Ball12} for further information. It is also known to be true for elliptic curve codes, see \cite{Wal96}. 
  
 Let $C$ be an $[n,k,d]$ linear code over $\f{q}$. The \emph{error distance} of any word $u\in\f{q}^n$ to $C$ is defined to be
$$d(u,C)=\min\{d(u,v)\,|\,v\in C\},$$
where $$d(u,v)=|\{i\,|\,u_{i}\neq v_{i},\,1\le i\le n\}|$$
is the Hamming distance between words $u$ and $v$.
Computing the error distance is essentially the ultimate decoding problem. Although there are decoding 
algorithms available for important codes such as Reed-Solomon codes and algebraic geometric codes, 
but these algorithms only work if the error 
distance $d(u, C)$ is small. If the error distance is large, then decoding becomes a problem of major difficulty. 
The extreme instance is the maximum error distance
\[
   \rho(C)=\max\{d(u,\, C)\,|\,u\in \f{q}^n\}
\]
which is called the \emph{covering radius} of $C$. 

The covering radius is perhaps the next most important quantity 
of a linear code, after the minimal distance. Covering radius of codes was studied extensively~\cite{CKMS85,CLS86,GS85,HKM78,MCL84,OST99}. 
There are very few families of codes with known covering radius, e.g., Reed-Solomon codes, the first order Reed-Muller code $RM(1,m)$ with even $m$, etc. For the first order Reed-Muller code $RM(1,m)$ with odd $m$, to determine the covering radius is already very difficult and wide open~\cite{Hou93}. A recent breakthrough is due to Schmidt~\cite{Sch19}.
Even for the projective Reed-Solomon code (which is an MDS code), the exact covering radius is surprisingly unknown, see~\cite{ZWK19} for the discussion.

If the distance from a word to the code achieves the covering radius of the code, then the word is called a \emph{deep hole} of the code. 
Deciding deep holes of a given code is an extreme instance of decoding. It is much harder than the covering radius problem, even for affine RS codes. The deep hole problem for Reed-Solomon codes was studied in~\cite{CMP11,CM07,Kaipa17,KW15,LW08,WL08,LZ15,Liao11,WH12,ZFL13,ZWK19,ZCL16}. 
For Reed-Solomon codes of length $n$ much smaller than $q$, deciding if a given word is a deep hole is equivalent to a 
general subset sum problem, which is NP-hard. For Reed-Solomon codes of length $n$ close to $q$, 
the deep hole problem 
can be solved if one assumes MDS conjecture or the rational normal curve conjecture in finite geometry, see \cite{ZWK19}. 
In summary, the deep hole problem is expected to be well structured for long Reed-Solomon codes,  
but no structure for short Reed-Solomon codes. 

In this paper, we will study deep holes of elliptic curve codes. 
For the definition and basics of elliptic curve codes, please see Section~\ref{Sec:Preliminaries}. 
Again, we expect that the deep hole problem is well structured for long elliptic curve codes, 
but no structure for short elliptic curve codes.  For this reason, we will mostly restrict to long elliptic curve codes 
in this paper. 

In practical applications, for codes of length $\leq q+2$, Reed-Solomon codes already play the best performance. For codes of length $n\geq q+3$, there are no non-trivial MDS codes by the MDS conjecture, and the next best thing would be near-MDS codes, i.e., $[n,k,d]$ codes with $d=n-k$. 
Long elliptic curve codes are known to be near-MDS and have the best parameters according to the Singleton bound. 
By the Hasse-Weil theorem, the length $n$ of an elliptic curve code $C$ is bounded above by $n \leq q + 2\sqrt{q}+1$. This 
significantly goes beyond the bound $n\leq q+1$ for Reed-Solom codes.  For good codes $C$ with length $n > q + 2\sqrt{q}+1$, 
one could use algebraic geometry codes of genus $g>1$ with many rational points. In this paper, we only consider the 
case $g=1$, which is already sufficiently interesting and difficult. 


From now on, we always assume that the finite field $\f{q}$ has odd characteristic throughout the rest of the paper. This is to avoid some small technical complications  and exceptions when the characteristic is two. 
In the statement of the following theorem, we focus on 
the functional elliptic curve $[n, k]$-code $C_{\mathcal{L}}(D, kO)$, see section 2 for precise definitions. 
With appropriate modification, the results hold for a general divisor $G$. We present the main result on minimun distance, 
covering radius, and deep holes of long elliptic curve codes.

\begin{thm}\label{thm:main}
    Let $E$ be an elliptic curve over $\f{q}$ with a rational point $O$, and $D\subset E(\f{q})\setminus\{O\}$ be a set of rational points with $n=|D|$. For $2\leq k\leq n-2$, let $C=C_{\mathcal{L}}(D, kO)$ be the functional elliptic curve $[n,k]$-code.  
Assume $n\geq q+3$ (the code is long). If any one of the following three conditions holds: 
	\begin{enumerate}
		\item[(1)] $n\geq 
		q+k$, or 
		\item[(2)] $q$ is a prime, or 
		\item[(3)] $k\leq \sqrt{q}$,
	\end{enumerate}
    then  we have the following results:
    \begin{enumerate}
    	\item[(i)] The minimum distance $d(C)=n-k$. 
    	\item[(ii)] The covering radius $\rho(C)=n-k-1$.
    	\item[(iii)] For any $P\in E(\f{q})\setminus D$, any vector $v\in C_{\mathcal{L}}(D, kO+P)\setminus C_{\mathcal{L}}(D, kO)$ is a deep hole of $C_{\mathcal{L}}(D, kO)$. 
    	\item[(iv)] If $k<n-2$, then the deep holes constructed in (iii) are all distinct and thus yield $(|E(\f{q})|-n)(q-1)q^{k}$ deep holes of $C_{\mathcal{L}}(D, kO)$. 
    \end{enumerate}
\end{thm}

\begin{rem} Since $n\geq q+3$, the minimun distance $d(C)=n-k$ is always true for elliptic codes. 
We need one of conditions (1)-(3) to insure that the covering radius can be proved to be $n-k-1$. These conditions can be removed if one assumes the MDS conjecture or 
if one simply assumes that the covering radius is $n-k-1$. Without one of these conditions, the covering radius is unknown and 
thus we cannot prove that the words constructed in (iii) are deep holes. 

For the boundary case $k=n-2$, under the conditions in the above theorem, the covering radius $\rho(C)=n-k-1=1$. So all vectors in $\f{q}^n\setminus C$ are deep holes of $C$. And hence, the code $C=C_{\mathcal{L}}(D, (n-2)O)$ totally has $(q^2-1)q^{n-2}$ deep holes.
\end{rem}

A natural question is if the construction in (iii) is complete, i.e., if there are other deep holes except those in (iii).  
If it is complete, then there will be exactly $(|E(\f{q})|-n)(q-1)q^{k}$ deep holes and the elliptic deep hole problem 
is solved. For short elliptic curve codes, the construction in (iii) will not be complete. However, 
we have the following completeness conjecture for sufficiently long elliptic curve codes, namely, when $D$ is the full set $E(\f{q})\setminus\{O\}$ and thus $n = |E(\f{q})|-1\geq q+3$. This conjecture is the elliptic code analogue of the Cheng-Murray 
conjecture \cite{CM07} for deep holes of Reed-Solomon codes. 

\begin{conj}\label{conj:completeness}
	 Let $E$ be an elliptic curve over $\f{q}$ with $|E(\f{q})|\geq q+4$. Take any rational point $O\in E(\f{q})$ and set $D= E(\f{q})\setminus\{O\}$.  Let $2\leq k\leq |E(\f{q})|-4$. Then,  
	 $C_{\mathcal{L}}(D, (k+1)O)\setminus C_{\mathcal{L}}(D, kO)$ is the set of 
	 all deep holes of $C_{\mathcal{L}}(D, kO)$.
	 
\end{conj}

The rest of this paper is organized as follows. In Section~\ref{Sec:Preliminaries}, we review the definition and basics of algebraic geometry (AG) codes. Regarding Reed-Solomon codes as AG codes constructed from the projective line, we give a new viewpoint of deep holes of Reed-Solomon codes to unify the two constructions in the previous works. In Section~\ref{Sec:coveringradius&deepholes}, we consider elliptic curve codes. We determine the covering radius, present deep holes and compute the syndromes of the deep holes under 
our suitable assumption. In Section~\ref{Sec:completenessofdeepholes}, we discuss the completeness of deep holes found in Section~\ref{Sec:coveringradius&deepholes}. The technique of group action of certain automorphisms of the elliptic curve are applied to yield more new deep holes if there is one new deep hole. The connection with finite geometry suggests a preliminary approach to the above completeness conjecture. 

\section{Preliminaries}\label{Sec:Preliminaries}

\subsection{Definitions}\label{Subsec:definitions}
In this subsection, We recall the definition and some basics of algebraic geometry codes. First fix some notations valid for the whole paper.

 \begin{itemize}
 	\item\emph{$\f{q}$ is a finite field of size $q$ where $q$ is an odd prime power. }
	\item\emph{
		$X/\f{q}$ is a geometrically irreducible smooth projective curve of genus $g$  over the finite field $\f{q}$
		with function field $\f{q}(X)$. }
	\item \emph{ $X(\f{q})$ is the set of all $\f{q}$-rational points on $X$.}
	\item\emph{ $D=\{P_{1},P_{2},\cdots,P_{n}\}$ is a proper
		subset of rational points $X(\f{q})$.}
	\item\emph{We also write $D=P_{1}+P_{2}+\cdots+P_{n}$.}
	\item\emph{ $G$ is a divisor of degree $k$
		($2g-2<k<n$) with $\mathrm{Supp}(G)\cap D=\emptyset$.} \end{itemize}

Let $V$  be a divisor on $X$. Denote by $\mathcal{L}(V)$ the $\f{q}$-vector space of all
rational functions $f\in \f{q}(X)$ with the principal divisor
$\mathrm{div}(f)\geqslant -V$, together with the zero function. It is well-known that $\mathcal{L}(V)$ is a 
finite dimensional vector space over $\f{q}$. And denote by $\Omega(V)$ the $\f{q}$-vector space of
all Weil differentials $\omega$ with divisor
$\mathrm{div}(\omega)\geqslant V$, together with the zero
differential (cf.~\cite{Stichtenoth}).

The residue AG code $C_{\Omega}(D, G)$ is defined to be the
image of the following residue map:
\begin{equation*}
\begin{array}{cccl}
\mathrm{res}: & \Omega(G-D) & \rightarrow & \f{q}^{n} \\
& \omega & \mapsto & (\mathrm{res}_{P_{1}}(\omega),\mathrm{res}_{P_{2}}(\omega),\cdots,\mathrm{res}_{P_{n}}(\omega))\ .
\end{array}
\end{equation*}
The code $C_{\Omega}(D, G)$ has parameters $[n, n-k-1+g, d\geq k-(2g-2)]$. And its dual
code, the functional AG code $C_{\mathcal{L}}(D, G)$ is defined to be the image of the following evaluation map:
\[
\mathrm{ev}: \mathcal{L}(G)\rightarrow \f{q}^{n};\, f\mapsto
(f(P_{1}),f(P_{2}),\cdots,f(P_{n})).\enspace
\]

As functions in $\mathcal{L}(G)$ have at most $k=\deg{G}$ different zeros, the minimum distance of $C_{\mathcal{L}}(D, G)$ is $d\geqslant n-k$. By the Riemann-Roch theorem, the functional AG code $C_{\mathcal{L}}(D, G)$ has parameters $[n, k-g+1, d\geqslant n-k]$. This together with the Singleton bound gives 
\[
k-(2g-2)\leq d(C_{\Omega}(D, G))\leq k-g+2
\]
and 
\[
 n-k\leq d(C_{\mathcal{L}}(D, G))\leq n-k+g.
\]

If $X=E$ is an elliptic curve over $\f{q}$, i.e., $g=1$, then $C_{\Omega}(D, G)$ has parameters $[n, n-k, d\geq k]$, 
$C_{\mathcal{L}}(D, G)$ has parameters $[n, k, d\geqslant n-k]$, and 
we only have the following two choices
for their minimum distance:
\[
d(C_{\Omega}(D, G))\in\{k,k+1\}\quad\mbox{and}\quad d(C_{\mathcal{L}}(D, G))\in\{n-k, n-k+1\}.
\]
It is easy to see that the two minimum distances take either  $\{k, n-k\}$ or $\{k+1, n-k+1\}$. In the first case, 
both $C_{\Omega}(D, G)$ and $C_{\mathcal{L}}(D, G)$ are near-MDS codes. In the second case, 
both $C_{\Omega}(D, G)$ and $C_{\mathcal{L}}(D, G)$ are MDS codes. It was shown that the MDS property is equivalent to certain general subset sum problem having no solution~\cite{chengqi}. So to determine the exact minimum distance of a general elliptic curve code is \textbf{NP}-hard under \textbf{RP}-reduction. However, non-trivial elliptic code of length $n\geq q+3$ 
are near-MDS by the MDS conjecture, i.e., the minimum distance take the smaller one. 
This suggests that long elliptic curve codes behave better. 
By using the Li-Wan sieve method~\cite{LW08}, the authors~\cite{LWZ15} improved the upper bound on the length of MDS elliptic curve codes without assuming the MDS conjecture. 

\begin{prop}[\cite{LWZ15}]\label{thm:mindist}
	Suppose that $n\geq (\frac 23+\epsilon)q$ and $q>\frac 4 {\epsilon^2}$,
	where $\epsilon $ is positive. There is a positive constant
	$C_{\epsilon}$ such that if
	$C_{\epsilon}\ln{q}<k<n-C_{\epsilon}\ln{q}$, then the $[n,k]$ elliptic curve code $C_{\mathcal{L}}(D, G)$	has minimum distance $n-k$ and hence is near-MDS.
\end{prop}

It is further conjectured in~\cite{LWZ15} that the above condition $n\geq (\frac 23+\epsilon)q$ can be improved to $n\geq (\frac 12+\epsilon)q$.  This conjecture has been proved in the case $3\leq k \leq \frac{q+1-2\sqrt{q}}{10}$ in the recent paper~\cite{HR21}.

\subsection{A new viewpoint of deep holes of Reed-Solomon codes}\label{Sec:deepholesofRScodes}
In this subsection, we give a new viewpoint of deep holes of Reed-Solomon codes regarded as algebraic geometry codes of genus zero. The advantage of this new viewpoint is that the two classes of deep holes of generalized Reed-Solomon codes are essentially the same. And this method extends immediately to elliptic curve codes and even more general AG codes.

Let $\f{q}(x)$ be the rational function field. Let $O$ be the infinite point with uniformizer $\frac{1}{x}$ and $P_a$ be the finite point with uniformizer $x-a$ for any $a\in\f{q}$. For any subset $D=\{a_1,a_2,\cdots,a_n\}\subset\f{q}$, denote the corresponding set of finite points also by $D=\{P_{a_1}, P_{a_2},
\cdots, P_{a_n}\}$. For any integer $1\leq k\leq n$, the Reed-Solomon (RS) code $RS(D,k)$ is defined to be $C_{\mathcal{L}}(D, (k-1)O)$. The dual code of RS code $RS(D,k)$ is the residue AG code $C_{\Omega}(D, (k-1)O)$. 
Both of these codes are MDS codes and their covering radius are easy to determine. 

For RS codes with odd $q$ and $k\geq \lfloor \frac{q-1}{2}\rfloor$, it was proved in~\cite{Kaipa17} that there are only two classes of deep holes. The first class of deep holes of RS code $RS(D,k)$ was given in~\cite{CM07} which corresponds to polynomials of degree $k$. In fact, these deep holes are vectors in $C_{\mathcal{L}}(D, kO)\setminus C_{\mathcal{L}}(D, (k-1)O)$.
The second class of deep holes of RS code $RS(D,k)$ was given first in~\cite{WH12} for $D=\f{q}^*$ and later in~\cite{ZFL13} for general $D\subsetneq \f{q}$ which corresponds to the rational functions $\{\frac{b}{x-a}\,|\,a\in \f{q}\setminus D, b\in \f{q}^*\}$. In fact, these deep holes are exactly the 
vectors in 
$$\bigcup_{a\in \f{q}\setminus D}C_{\mathcal{L}}(D, (k-1)O+P_a)\setminus C_{\mathcal{L}}(D, (k-1)O).$$

In the language of AG codes, the two classes of known deep holes can be unified as follows: for any $P\in \mathbb{P}^1(\f{q})\setminus D$, the vectors in  $C_{\mathcal{L}}(D, (k-1)O+P)\setminus C_{\mathcal{L}}(D, (k-1)O)$ are deep holes of the Reed-Solomon code $RS(D,k)$. In the case $D$ is the full set $\f{q}$, the only possibility for $P$ is $O$. In this case, the above construction of deep holes is conjectured to be complete in \cite{CM07}, which has been proved to be true if $q$ is a prime \cite{ZCL16} or if $k>(q-1)/2$ \cite{Kaipa17}, using results from finite geometry.

\section{Covering radius and deep holes}\label{Sec:coveringradius&deepholes}

Before we move on to address the deep hole problem for elliptic curve codes, we need to first understand the 
covering radius for elliptic curve codes. But this is already a difficult problem as seen below. 

\subsection{Covering radius of elliptic curve codes}
In this subsection, we study the covering radius of elliptic curve codes. 
Just like the minimal distance, in general, there are only two possible choices for the covering radius of elliptic curve codes.

\begin{lem}\label{lem:2optsofcoveringradius}
	Let $\f{q}$ be a finite field with $q$ elements. Let $E$ be an elliptic curve over $\f{q}$ with a rational point $O$, and $D\subset E(\f{q})\setminus\{O\}$ be a set of rational points with $n=|D|$. For $2\leq k\leq n-2$, let $C=C_{\Omega}(D, kO)$ or $C=C_{\mathcal{L}}(D, kO)$. Then the covering radius of $C$ equals either $n-\dim(C)-1$ or $n-\dim(C)$.
\end{lem}
\begin{proof}
	Denote $k^\perp=n-\dim(C)$. Let $H\in\f{q}^{k^\perp\times n}$ be any parity-check matrix for the linear code $C$. Then the covering radius $\rho$ of $C$ is the smallest positive integer $\rho$ such that any vector $w\in \f{q}^{k^\perp}$ can be written as a linear combination of some $\rho$ columns of $H$. As $H$ is of full rank, we have 
	\[
	\rho\leq k^\perp.
	\]
	On the other hand,  we consider any vector 
	\[	
	v\in \begin{cases}
	C_{\Omega}(D, (k-1)O)\setminus C_{\Omega}(D, kO),&\mbox{if $C=C_{\Omega}(D, kO)$};\\
	C_{\mathcal{L}}(D, (k+1)O)\setminus C_{\mathcal{L}}(D, kO),& \mbox{if $C=C_{\mathcal{L}}(D, kO)$}.
	\end{cases}
	\]
	By the Singleton bound, 
	$$k^{\perp} = n -\dim(C) \leq d(C) \leq n-\dim(C) +1 = k^\perp +1,$$
	$$k^\perp -1 \leq d(C_{\Omega}(D, (k-1)O)) \leq k^\perp, $$
	$$k^\perp -1 \leq d(C_{\mathcal{L}}(D, (k+1)O)) \leq k^\perp, $$we deduce that 
	\[
	\rho\geq d(v, C)\geq \begin{cases}
	\min(d(C), d(C_{\Omega}(D, (k-1)O)))&\mbox{if $C=C_{\Omega}(D, kO)$}\\
	\min(d(C), d(C_{\mathcal{L}}(D, (k+1)O)))& \mbox{if $C=C_{\mathcal{L}}(D, kO)$}
	\end{cases}\geq k^\perp -1.
	\]
	So $\rho\in\{k^\perp, k^\perp-1\}=\{n-\dim(C)-1, n-\dim(C)\}$.
\end{proof}

\begin{rem}
For short elliptic curve codes, to determine the minimum distance is already \textbf{NP}-hard~\cite{chengqi} under \textbf{RP}-reduction. To determine the covering radius is even harder, which not only depends on the MDS property but also on certain extendability of MDS or near-MDS codes. 
	For instance, if $C$ is MDS, then $k^\perp=d(C)-1$. The covering radius of $C$ can still take any one of the two choices $\{n-\dim(C)-1, n-\dim(C)\}$. For the covering radius of $C$ to be $n-\dim(C)$, it is equivalent to that there is a vector $v\in\f{q}^n\setminus C$ such that $C\oplus \f{q}v$ is MDS. Even for $v\in C_{\mathcal{L}}(D, kO+P)\setminus C_{\mathcal{L}}(D, kO)$, the problem is already hard which is equivalent to certain subset sum problem.
\end{rem}	
	
However, for long elliptic curve codes of length $n\geq q+3$, the problem becomes easier, at least under the MDS conjecture. And long AG codes are preferred in applications anyway.

 Recall that an $[n,k,d]$ linear code is called {\it optimal} if there does not exist $[n',k,d]$ linear code with $n'<n$.

\begin{thm}\label{thm:coveringradius}
	 Let $\f{q}$ be a finite field with $q$ elements. Let $E$ be an elliptic curve over $\f{q}$ which has at least $q+4$ rational points. Let $O\in E(\f{q})$ be a rational point on $E$ and $D\subset E(\f{q})\setminus\{O\}$ be a set of rational points with $n=|D|\geq q+3$. For $2\leq k\leq n-2$, let $C=C_{\Omega}(D, kO)$ or $C=C_{\mathcal{L}}(D, kO)$. 
	 The minimun distance of $C$ is given by $d(C)= n-\dim(C)$. If we further assume that the MDS conjecture holds 
	 for all $[n-1, \dim(C)]$-codes over $\f{q}$, then 
	 the covering radius of $C$ is given by $\rho(C)=d(C)-1 = n-\dim(C)-1$.
\end{thm}
\begin{proof}
	The MDS conjecture is known to be true for the elliptic code $C$, see \cite{Wal96}. 
	Since $n\geq q+3$, this implies that 
	the code $C$ is not MDS, and hence must be near-MDS with parameters $[n,\dim(C),d(C)=n-\dim(C)]$. In particular, the minimun distance is $d(C)= n -\dim(C)$.  Now $n-1\geq q+2$ and $q$ is odd. By the MDS conjecture for $[n-1, \dim(C)]$-codes, we deduce that there is no MDS code with parameters $[n-1,\dim(C),d(C)=n-\dim(C)]$. So the code $C$ is optimal. By~\cite[Corollary~8.1]{Jan90}, we have the following bound on the covering radius 
	\[
	\rho(C)\leq d(C)-\lceil \frac{d(C)}{q^{\dim(C)}}\rceil.
	\]	
	Since $d(C) = n -\dim(C)>0$, it follows that 
	\[
	\rho(C)\leq d(C)-1.
	\]
	On the other hand, by Lemma~\ref{lem:2optsofcoveringradius}, $\rho(C)\in\{d(C)-1, d(C)\}.$
	We conclude  that $\rho(C)=d(C)-1.$
	
\end{proof}	

Now, any non-trivial MDS code $C$ of dimension $k\geq 3$ over the finite field $\f{q}$ has length $n\leq q+k-2$ by \cite[Chapter~11, Theorem~11]{Mac}. The MDS conjecture holds for prime fields~\cite{Ball12}, and also for general $q$ with $k\leq \sqrt{q}$~\cite{Segre55,BL20}. This means that under one of the conditions (1)-(3) in Theorem~\ref{thm:main}, 
the MDS conjecture holds for all $[n, k]$-codes and $[n-1, k]$-codes over $\f{q}$. 
As a consequence, we obtain the conclusions (i) and (ii) of Theorem~\ref{thm:main} from Theorem~\ref{thm:coveringradius}. 

\subsection{Deep holes of elliptic curve codes and their syndromes} 
In this subsection, we first prove Theorem~\ref{thm:main}(iii)-(iv). In order to study further geometry of deep holes, 
we will later focus on residue elliptic curve codes since they have an explicit parity-check matrix of the form (\ref{equ:parity-checkmat}). 
The residue and functional algebraic geometry codes can be represented by each other~(cf.~\cite[Proposition~2.2.10]{Stichtenoth}). 
Thus, in principal, it is sufficient to consider the residue elliptic curve codes.

\begin{lem}\label{thm:orderofpole}
	For any rational point $P\in E(\f{q})\setminus D$ and any $f\in \mathcal{L}(kO+P)\setminus\mathcal{L}(kO)$, we have 
	\begin{itemize}
		\item[1.] If $P\neq O$, then $P$ is a simple pole of $f$.
		\item[2.] If $P=O$, then $O$ is a pole of $f$ of order $k+1.$
	\end{itemize}
\end{lem}
\begin{flushleft}
	{\it Proof of Theorem~\ref{thm:main}(iii)-(iv).} For any rational point $P\in E(\f{q})\setminus D$ and any vector $v\in C_{\mathcal{L}}(D, kO+P)\setminus C$, we have 
	\[
	n-k-1=\rho(C)\geq d(v, C)\geq \min (d(C), d(C_{\mathcal{L}}(D, kO+P)))\geq n-(k+1).
	\]
	So $d(v, C)=n-k-1$. That is, the vector $v$ is a deep hole. This proves Theorem~\ref{thm:main}(iii). 
	\end{flushleft}

  Next, we prove Theorem~\ref{thm:main}(iv). For any two distinct rational points $P, Q\in E(\f{q})\setminus D$, since functions in $\mathcal{L}(kO+P+Q)$ have at most $k+2<n$ zeros, the evaluation map $\mathrm{ev}\,:\,\mathcal{L}(kO+P+Q)\rightarrow \f{q}^n$ is injective. The codes $C_{\mathcal{L}}(D, kO+P)$ and $C_{\mathcal{L}}(D, kO+Q)$ are two sub-codes of the code $C_{\mathcal{L}}(D, kO+P+Q)$.
 Now, if there exsit $f\in \mathcal{L}(kO+P)\setminus \mathcal{L}(kO)$ and $g\in \mathcal{L}(kO+Q)\setminus \mathcal{L}(kO)$ such that $\mathrm{ev}(f)=\mathrm{ev}(g)$, then $f-g$ regarded as a function in $\mathcal{L}(kO+P+Q)$ satisfies
 \[
 \mathrm{ev}(f-g)=0.
 \]
  Since the evaluation map $\mathrm{ev}\,:\,\mathcal{L}(kO+P+Q)\rightarrow \f{q}^n$ is injective, we have $f-g=0$, i.e., $f=g$ in $\mathcal{L}(kO+P+Q)$.
  This is impossible according to Lemma~\ref{thm:orderofpole} by comparing the orders of poles $P$ and $Q$.
 So the sets $C_{\mathcal{L}}(D, kO+P)\setminus C$ and $C_{\mathcal{L}}(D, kO+Q)\setminus C$ are disjoint for any distinct rational points $P, Q\in E(\f{q})\setminus D$. 
  
  For any rational point $P\in E(\f{q})\setminus D$, we have 
  \[
  |C_{\mathcal{L}}(D, kO+P)\setminus C|=q^{k+1}-q^k=(q-1)q^k.
  \]
  According to the above disjointness, there are totally 
  \[
  |E(\f{q})\setminus D|(q-1)q^k=(|E(\f{q})|-n)(q-1)q^k
  \]
  deep holes provided by the theorem.
  \begin{flushright}
  	$\qed$
  \end{flushright}
  
\begin{rem}
    For the case $k= n-2$, let $P, Q\in E(\f{q})\setminus D$ be two distinct rational points. Then the spaces $C_{\mathcal{L}}(D, kO+P)=C_{\mathcal{L}}(D, kO+Q)$ if and only if the divisor $D-kO-P-Q$ is principal.
\end{rem}

In the rest of this subsection, we focus on the residue elliptic codes. We will compute the syndromes of deep holes. 
This will be useful to connect to finite geometry. 
Because we only consider elliptic curves over finite fields of odd characteristic, we may assume that the elliptic curve is given by $y^2=x^3+sx+t$ ($s,t\in\f{q}$) together with the infinity point $O$. 
Let 
$$D=\{P_i=(\alpha_i, \beta_i)\mid i=1,2,\cdots,n\}\subset E(\f{q})\setminus\{O\}\}$$
be a set of rational points on $E$ of size $n=|D|$. Let $C=C_{\Omega}(D, kO)$ be the residue elliptic curve code. Then the dual code of $C$ is $C^\perp=C_{\mathcal{L}}(D, kO)$. The Riemann-Roch space $\mathcal{L}(kO)$ has a basis $\{x^iy^j\mid i\in\mathbb{Z}_{\geq 0},\,j\in\{0,1\},\,2i+3j\leq k\}$. So we may choose a parity-check matrix of $C$ as follows:
\begin{equation}\label{equ:parity-checkmat}
	H(k)=\left(
\begin{array}{cccc}
1 & 1 & \cdots & 1  \\
\alpha_1 & \alpha_2 & \cdots & \alpha_n  \\
\vdots & \vdots & \ddots & \vdots  \\
\alpha_1^{\lfloor\frac k2\rfloor} & \alpha_2^{\lfloor\frac k2\rfloor} & \cdots & \alpha_n^{\lfloor\frac k2\rfloor} \\
\beta_1 & \beta_2 & \cdots & \beta_n  \\
\alpha_1\beta_1 & \alpha_2\beta_2 & \cdots & \alpha_n\beta_n  \\
\vdots & \vdots & \ddots & \vdots  \\
\alpha_1^{\lfloor\frac {k-3}2\rfloor}\beta_1 & \alpha_2^{\lfloor\frac {k-3}2\rfloor}\beta_2 & \cdots & \alpha_n^{\lfloor\frac {k-3}2\rfloor}\beta_n \\
\end{array}
\right).
\end{equation}

\begin{thm}\label{thm:deephole&syn}
	Notations as above. Suppose the $[n, n-k]$-code $C=C_{\Omega}(D, kO)$ has covering radius $\rho(C)=n-\dim(C)-1=k-1$. Let $P=(\alpha, \beta)\in E(\f{q})\setminus D$ be any rational point on the elliptic curve $E$. Then we have
	\begin{enumerate}
		\item Any vector $v\in C_{\Omega}(D, kO-P)\setminus C$ is a deep hole of $C$.
		\item If $P=O$, then the syndrome of $v$ is
		      \[
		      H(k)v^T=\begin{cases}
		      (0,\cdots,0,\sum_{i=1}^{n}\alpha_i^{\frac k2}v_i,0,\cdots, 0)^T,&\mbox{if $k$ is even;}\\
		      (0,\cdots,0,\sum_{i=1}^{n}\alpha_i^{\frac {k-3}2}\beta_iv_i)^T,&\mbox{if $k$ is odd.}\\
		      \end{cases}
		      \]
		\item If $P\neq O$, then the syndrome of $v$ is
		      \[
		      H(k)v^T=b(1,\alpha, \cdots, \alpha^{\lfloor\frac k2\rfloor}, \beta, \beta\alpha,\cdots, \beta\alpha^{\lfloor\frac {k-3}2\rfloor})^T,
		      \]
		      where $b=\sum_{i=1}^{n}v_i \not=0.$
	\end{enumerate}
\end{thm}
\begin{proof}
Since $d(C) \geq n-\dim(C) = n-(n-k)=k$, 
the first statement follows from 
\[
  k-1=\rho(C)\geq d(v, C)\geq \min(d(C), d(C_{\Omega}(D, kO-P)))\geq k-1.
\] 	
	
 Next, we compute the syndrome $H(k)v^T$ by separating two cases: $P=O$ and $P\in E(\f{q})\setminus (D\cup \{O\})$.

For the case $P=O$, any vector $v\in C_{\Omega}(D, (k-1)O)\setminus C_{\Omega}(D, kO))$ can be killed by vectors in $C_{\mathcal{L}}(D, (k-1)O)$ but not by vectors in $C_{\mathcal{L}}(D, kO)\setminus C_{\mathcal{L}}(D, (k-1)O)$. So $H(k-1)v^T=0.$ Hence, the syndrome equals
\[
H(k)v^T=\begin{cases}
(0,\cdots,0,\sum_{i=1}^{n}\alpha_i^{\frac k2}v_i,0,\cdots, 0)^T,&\mbox{if $k$ is even;}\\
(0,\cdots,0,\sum_{i=1}^{n}\alpha_i^{\frac {k-3}2}\beta_iv_i)^T,&\mbox{if $k$ is odd.}\\
\end{cases}
\]

For the case $P=(\alpha, \beta)\in E(\f{q})\setminus (D\cup \{O\})$, similarly, any vector $v\in C_{\Omega}(D, kO-P)\setminus C_{\Omega}(D, kO)$ is killed by vectors in $C_{\mathcal{L}}(D, kO-P)$ but not by vectors in $C_{\mathcal{L}}(D, kO)\setminus C_{\mathcal{L}}(D, kO-P)$. It is clear that 
$$\{x-\alpha, (x-\alpha)x, \cdots, (x-\alpha)x^{\lfloor\frac k2\rfloor}, y-\beta, x(y-\beta), \cdots, x^{\lfloor\frac {k-3}2\rfloor}(y-\beta)\}$$ forms a basis for the Riemann-Roch space $\mathcal{L}(kO-P)$. So we have{\small
$$\left(
\begin{array}{cccc}
\alpha_1-\alpha & \alpha_2-\alpha & \cdots & \alpha_n-\alpha  \\
(\alpha_1-\alpha)\alpha_1 & (\alpha_2-\alpha)\alpha_2 & \cdots & (\alpha_n-\alpha)\alpha_n  \\
\vdots & \vdots & \ddots & \vdots  \\
(\alpha_1-\alpha)\alpha_1^{\lfloor\frac k2\rfloor} & (\alpha_2-\alpha)\alpha_2^{\lfloor\frac k2\rfloor} & \cdots & (\alpha_n-\alpha)\alpha_n^{\lfloor\frac k2\rfloor} \\
\beta_1-\beta & \beta_2-\beta & \cdots & \beta_n-\beta  \\
\alpha_1(\beta_1-\beta) & \alpha_2(\beta_2-\beta) & \cdots & \alpha_n(\beta_n-\beta)  \\
\alpha_1^2(\beta_1-\beta) & \alpha_2^2(\beta_2-\beta) & \cdots & \alpha_n^2(\beta_n-\beta)  \\
\vdots & \vdots & \ddots & \vdots  \\
\alpha_1^{\lfloor\frac {k-3}2\rfloor}(\beta_1-\beta) & \alpha_2^{\lfloor\frac {k-3}2\rfloor}(\beta_2-\beta) & \cdots & \alpha_n^{\lfloor\frac {k-3}2\rfloor}(\beta_n-\beta)  \\
\end{array}
\right)v^T=0.$$}
Let $b=\sum_{i=1}^{n}v_i$. Since the vector $v$ can not be killed by vectors in $C_{\mathcal{L}}(D, kO)\setminus C_{\mathcal{L}}(D, kO-P)$, we have $b\neq 0.$ So we have{\small
$$\left(
\begin{array}{cccc}
1&1&\cdots&1\\
\alpha_1-\alpha & \alpha_2-\alpha & \cdots & \alpha_n-\alpha  \\
(\alpha_1-\alpha)\alpha_1 & (\alpha_2-\alpha)\alpha_2 & \cdots & (\alpha_n-\alpha)\alpha_n  \\
\vdots & \vdots & \ddots & \vdots  \\
(\alpha_1-\alpha)\alpha_1^{\lfloor\frac k2\rfloor} & (\alpha_2-\alpha)\alpha_2^{\lfloor\frac k2\rfloor} & \cdots & (\alpha_n-\alpha)\alpha_n^{\lfloor\frac k2\rfloor} \\
\beta_1-\beta & \beta_2-\beta & \cdots & \beta_n-\beta  \\
\alpha_1(\beta_1-\beta) & \alpha_2(\beta_2-\beta) & \cdots & \alpha_n(\beta_n-\beta)  \\
\alpha_1^2(\beta_1-\beta) & \alpha_2^2(\beta_2-\beta) & \cdots & \alpha_n^2(\beta_n-\beta)  \\
\vdots & \vdots & \ddots & \vdots  \\
\alpha_1^{\lfloor\frac {k-3}2\rfloor}(\beta_1-\beta) & \alpha_2^{\lfloor\frac {k-3}2\rfloor}(\beta_2-\beta) & \cdots & \alpha_n^{\lfloor\frac {k-3}2\rfloor}(\beta_n-\beta)  \\
\end{array}
\right)v^T=\left(\begin{array}{l}
b\\
0\\
\vdots\\
0
\end{array}
\right).$$}
By adding the second row by $\alpha$-times of the first row, then adding the third row by $\alpha$-times of the new second row, and so on, we obtain{\small
$$\left(
\begin{array}{cccc}
1&1&\cdots&1\\
\alpha_1 & \alpha_2 & \cdots & \alpha_n  \\
\vdots & \vdots & \ddots & \vdots  \\
\alpha_1^{\lfloor\frac k2\rfloor} & \alpha_2^{\lfloor\frac k2\rfloor} & \cdots & \alpha_n^{\lfloor\frac k2\rfloor} \\
\beta_1-\beta & \beta_2-\beta & \cdots & \beta_n-\beta  \\
\alpha_1(\beta_1-\beta) & \alpha_2(\beta_2-\beta) & \cdots & \alpha_n(\beta_n-\beta)  \\
\alpha_1^2(\beta_1-\beta) & \alpha_2^2(\beta_2-\beta) & \cdots & \alpha_n^2(\beta_n-\beta)  \\
\vdots & \vdots & \ddots & \vdots  \\
\alpha_1^{\lfloor\frac {k-3}2\rfloor}(\beta_1-\beta) & \alpha_2^{\lfloor\frac {k-3}2\rfloor}(\beta_2-\beta) & \cdots & \alpha_n^{\lfloor\frac {k-3}2\rfloor}(\beta_n-\beta)  \\
\end{array}
\right)v^T=\left(\begin{array}{l}
b\\
b\alpha\\
\vdots\\
b\alpha^{\lfloor\frac k2\rfloor}\\
0\\
\vdots\\
0
\end{array}
\right).$$}
Now, by adding $\beta$ times of the first $\lfloor\frac {k-3}2\rfloor +1$ rows to the lower part in the above equation, we have
\[
\left(
\begin{array}{cccc}
1 & 1 & \cdots & 1  \\
\alpha_1 & \alpha_2 & \cdots & \alpha_n  \\
\vdots & \vdots & \ddots & \vdots  \\
\alpha_1^{\lfloor\frac k2\rfloor} & \alpha_2^{\lfloor\frac k2\rfloor} & \cdots & \alpha_n^{\lfloor\frac k2\rfloor} \\
\beta_1 & \beta_2 & \cdots & \beta_n  \\
\alpha_1\beta_1 & \alpha_2\beta_2 & \cdots & \alpha_n\beta_n  \\
\vdots & \vdots & \ddots & \vdots  \\
\alpha_1^{\lfloor\frac {k-3}2\rfloor}\beta_1 & \alpha_2^{\lfloor\frac {k-3}2\rfloor}\beta_2 & \cdots & \alpha_n^{\lfloor\frac {k-3}2\rfloor}\beta_n \\
\end{array}
\right)v^T=\left(\begin{array}{l}
b\\
b\alpha\\
\vdots\\
b\alpha^{\lfloor\frac k2\rfloor}\\
b\beta\\
b\beta \alpha\\
\vdots\\
b\beta \alpha^{\lfloor\frac {k-3}2\rfloor}
\end{array}
\right).
\]
That is, 
\[
H(k)v^T=b(1,\alpha, \cdots, \alpha^{\lfloor\frac k2\rfloor}, \beta, \beta\alpha,\cdots, \beta\alpha^{\lfloor\frac {k-3}2\rfloor})^T.
\]
\end{proof}

In the above theorem, one checks that 
$$|C_{\Omega}(D, kO-P)\setminus C| = q^{n-k+1} - q^{n-k} = (q-1)q^{n-k}.$$
If $k\geq 3$, the syndrome formula implies that the union 
$$ \bigcup_{P\in E(\f{q})\setminus D} C_{\Omega}(D, kO-P)\setminus C$$
is a disjoint uinion. This disjoint union gives $(|E(\f{q})|-n)(q-1)q^{n-k}$ deep holes of $C$. 
This proves the analogue of Theorem \ref{thm:main} for the residue elliptic code $C_{\Omega}(D, kO)$. 



 \begin{rem} 
  It is natural to ask if the deep holes given in the above theorem form all the deep holes. 
  The answer is expected to be yes for long elliptic curve codes, but it is going to be difficult to prove. In next section, we use finite geometry to discuss the completeness of deep holes in the above theorem and use the automorphism technique to obtain more ones if there is any new deep hole. 
\end{rem}

\section{On the completeness of deep holes}\label{Sec:completenessofdeepholes}
\subsection{Automorphisms of elliptic curves and deep holes of elliptic curve codes}

In this subsection, we use a Hamming distance-preserving subgroup of automorphism group of the linear code $C$ to 
construct more deep holes if there is any new deep hole. This idea was used in~\cite{ZWK19}. 

Let $\aut(E/\f{q})$ be the $\f{q}$-automorphism group of $E$ as an elliptic curve and  let $\aut_{D,G}(E/\f{q})$ be its subgroup fixing $D$ and $G$, respectively. That is, 
$$\aut_{D,G}(E/\f{q})=\{\sigma\in\aut(E/\f{q})\,|\,\sigma(D)=D,\,\sigma(G)=G\}.$$For any linear code $C$, denote by $\paut(C)$ the permutation automorphism group of $C$ whose elements are not only permutations of coordinates but also automorphisms of $C$.

\begin{lem}
	Let $C=C_{\Omega}(D,G)$ or $C=C_{\mathcal{L}}(D,G)$ be the algebraic geometry code constructed from the elliptic curve $E$. There is a homomorphism $\rho\,:\, \aut_{D,G}(E/\f{q})\rightarrow \paut(C)$.
\end{lem}
\begin{proof}
	We only prove the statement for $C=C_{\mathcal{L}}(D,G)$. The proof for the case $C=C_{\Omega}(D,G)$ is the same. Let $\aut_{D,G}(\f{q}(E)/\f{q})$ be the $\f{q}$-automorphism group of the elliptic function field $\f{q}(E)$ whose elements fix $D$ and $G$. There is a homomorphism 
	\[
	\rho_1\,:\,\aut_{D,G}(E/\f{q})\rightarrow \aut_{D,G}(\f{q}(E)/\f{q})
	\]
	defined by as follows: for any $T\in \aut_{D,G}(E/\f{q})$, and $f\in \f{q}(E)$, $\rho_1(T)(f)=T^*(f)$ is the pull-back of $f$ which defined by $T^*(f)(P)=f(T^{-1}(P))$ for any $P\in E(\f{q})$.
	
	Next, we show that for any $T\in \aut_{D,G}(E/\f{q})$, $\rho_1(T)\in \aut(\mathcal{L}(G))$. For any $f\in\mathcal{L}(G)$, we have $\mathrm{div}(f)+G\geq 0.$ So $T^{-1}(\mathrm{div}(f))+T^{-1}(G)\geq 0.$ Since $T^{-1}(G)=G$ and $T^{-1}(\mathrm{div}(f))=\mathrm{div}(T^*(f))$, we have $\mathrm{div}(T^*(f))+G\geq 0$. That is, $T^*(f)\in \mathcal{L}(G).$ 
	
	Since the map $\mathrm{ev}$ is an isomorphism from $\mathcal{L}(G)$ to $C$, we define $\rho=\mathrm{ev}\circ \rho_1\circ \mathrm{ev}^{-1}\,:\,C\rightarrow C$. It is obvious that $\rho$ is an automorphism of $C$.
	To finish the proof, we need to show that for any $T\in \aut_{D,G}(E/\f{q})$, $\rho(T)$ is a permutation of coordinates. Indeed, for any $f\in\mathcal{L}(G)$, we have
	\begin{align*}
	  &\rho(T)(f(P_1), f(P_2), \cdots, f(P_n))=\mathrm{ev}\circ \rho_1(T)\circ \mathrm{ev}^{-1}(f(P_1), f(P_2), \cdots, f(P_n))\\
	  =&\mathrm{ev}( \rho_1(T)(f))=\mathrm{ev}(T^*(f))=(T^*(f)(P_1), T^*(f)(P_2), \cdots, T^*(f)(P_n))\\
	  =&(f(T^{-1}(P_1)), f(T^{-1}(P_2)),\cdots, f(T^{-1}(P_n))).
	\end{align*}
	
\end{proof}	

Note that for any $T\in \aut_{D,G}(E/\f{q})$, as a permutation of coordinates, the map $\rho(T)$ can be extended to the whole space $\f{q}^n$ which is still denoted by $\rho(T)$.
\begin{prop}\label{thm:automorphism&deephole}
	Let $C=C_{\Omega}(D,G)$ or $C=C_{\mathcal{L}}(D,G)$ be the algebraic geometry code constructed from the elliptic curve $E$. If the word $v$ is a deep hole of $C$, then so is the word $\rho(T)(v)$ for any $T\in \aut_{D,G}(E/\f{q})$.
\end{prop}
\begin{proof}
	We have seen that the map $\rho(T)\,:\,\f{q}^n\rightarrow \f{q}^n$ preserves the Hamming distance and is an automorphism of $C$ if restricted to the linear subspace $C$. So $$d(\rho(T)(v), C)=d(v, \rho(T)^{-1}C)=d(v, C).$$ Hence, the word $\rho(T)(v)$ is a deep hole of $C$.
\end{proof}	

\begin{rem}
	Since $T(P)\in E(\f{q})\setminus D$ for any $T\in \aut_{D,G}(E/\f{q})$ and $P\in E(\f{q})\setminus D$, the deep holes found in Theorem~\ref{thm:deephole&syn} are invariant under the action in Proposition~\ref{thm:automorphism&deephole}. 
	
	If any new deep hole except those in Theorem~\ref{thm:deephole&syn} was found, then its orbit under the action of $\rho(\aut_{D,G}(E/\f{q}))$ would provide new ones. So it is interesting to find new deep hole not of the form in Theorem~\ref{thm:deephole&syn}. We will see this is already very hard for small $n-k$ in the next subsection.
\end{rem}

\subsection{Deep holes of elliptic curve codes and finite geometry}\label{Sec:geometryofdeepholes}
In this subsection, we discuss the geometry of deep holes of elliptic curve codes.

\begin{definition}
	An $n$-track in $PG(k-1,\f{q})$, the projective $k-1$-dimensional space over the finite field $\f{q}$, is a set $\mathcal{T}$ of $n$ points which satisfies the following two conditions: 
	\begin{enumerate}
		\item[(i)] Any $k-1$ points in $\mathcal{T}$ are linearly independent as vectors in $\f{q}^k$;
		\item[(ii)] There exists a hyperplane passing through some $k$ points in $\mathcal{T}$.
	\end{enumerate} 
\end{definition}

\begin{definition}
	An $(n;k)$-set $\mathcal{A}$ in $PG(k-1,\f{q})$ is an $n$-track with an extra condition: any $k+1$ points of $\mathcal{A}$ can linearly generate $PG(k-1,\f{q})$.
\end{definition}

The following proposition give the structure of long tracks.

\begin{prop}\cite[Theorem~3.4]{DL94}\label{longtracks}
	If $n>q+k$, then any $n$-track in $PG(k-1,\f{q})$ is an $(n;k)$-set. 
\end{prop}

\begin{definition}
	An $(n;k)$-set $\mathcal{A}$ in $PG(k-1,\f{q})$ is called complete if there is no $(n+1;k)$-set in $PG(k-1,\f{q})$ containing $\mathcal{A}$ as a subset.
\end{definition}

\begin{definition}
	An $(n;k)$-set $\mathcal{A}$ in $PG(k-1,\f{q})$ is called extendable if there is some $n+1$-track in $PG(k-1,\f{q})$ containing $\mathcal{A}$ as a subset. Otherwise, we call it non-extendable.
\end{definition}

Note that when $n>q+k$, all $n$-tracks in $PG(k-1,\f{q})$ are $(n;k)$-set. 
So in this case, ``complete" and ``non-extendable" are the same thing.

An important class of long $(n;k)$-sets is constructed from elliptic curves, i.e., columns of $H(k)$. We can rearrange the rows of $H(k)$ such that $2i+3j$ of the corresponding row defined by $x^iy^j$ is in the increasing order. Define the map
\begin{align*}
\phi_k\,:\,&E(\f{q})\setminus \{O\}\rightarrow PG(k-1,\f{q})\\
         &(x,y)\mapsto \begin{cases}
         (1,x,y,x^2,xy,x^3,\cdots,x^{\frac{k}{2}-2}y,x^{\frac{k}{2}})^T&\mbox{if $k$ is even;}\\
         (1,x,y,x^2,xy,x^3,\cdots,x^{\frac{k-1}{2}}, x^{\frac{k-3}{2}}y)^T&\mbox{if $k$ is odd,}\\
         \end{cases}\notag
\end{align*}
and $\phi_k(O)=(0,0,\cdots,0,1)^T\in PG(k-1,\f{q})$.
If $n=|D|\geq q+1$, then the code $C_\Omega(D\cup\{P\}, kO)$ has minimum distance $k$ for any rational point $P\in E(\f{q})\setminus D$ by MDS conjecture (for a proof of MDS conjecture for elliptic curve codes, we refer to \cite{Wal96}). Note that the matrix $[\phi_k(P_1), \cdots, \phi_k(P_n), \phi_k(P)]$ is a parity-check matrix of the code $C_\Omega(D\cup\{P\}, kO)$, so the vectors $\phi_k(P_1), \phi_k(P_2),\cdots, \phi_k(P_n)$ and $\phi_k(P)$ form an $(n+1; k)$-set for any rational point $P\in E(\f{q})\setminus D$.
\begin{prop}\label{thm:geometryofdeephole}
	Suppose the residue elliptic curve code $C=C_{\Omega}(D, kO)$ has covering radius $\rho=k-1$. Let $H=(h_1,h_2,\cdots, h_n)\in\f{q}^{k\times n}$ be a parity-check matrix of $C$. The vector $v$ is a deep hole of $C$ if and only if vectors $h_1,h_2,\cdots, h_n$ and $Hv^T$ form an $n+1$-track in $PG(k-1,\f{q})$.
\end{prop}
\begin{proof}
	First, since the code $C=C_{\Omega}(D, kO)$ has minimum distance $\geq k$, any $k-1$ columns of the parity-check matrix $H$ are linearly independent.
	
	Secondly, since $C$ has covering radius $\rho=k-1$, we have that the vector $v$ is a deep hole of $C$ if and only the syndrome $Hv^T$ can not be written as any linear combination of any $\leq k-2$ columns of $H$.
	
	So if the vector $v$ is a deep hole of $C$, then 
	\begin{itemize}
		\item[(i)] any $k-1$ vectors from $\{h_1,h_2,\cdots, h_n, Hv^T\}$ are linearly independent;
		\item[(ii)]  there exists a hyperplane passing through $Hv^T$ and some $k-1$ vectors from $\{h_1,h_2,\cdots, h_n\}$, since $d(v,C)=k-1$.
	\end{itemize}
	Hence $\{h_1,h_2,\cdots, h_n, Hv^T\}$ forms an $n+1$-track in $PG(k-1,\f{q})$. 
	
	Conversely, if  $\{h_1,h_2,\cdots, h_n, Hv^T\}$ is an $n+1$-track in $PG(k-1,\f{q})$, then any $k-1$ vectors from $\{h_1,h_2,\cdots, h_n, Hv^T\}$ are linearly independent. So the syndrome $Hv^T$ can not be written as any linear combination of any $\leq k-2$ columns of $H$. Hence 
	the vector $v$ is a deep hole of $C$.

\end{proof}	

The above proposition characterizes the geometry of deep holes of the residue elliptic curve code $C=C_{\Omega}(D, kO)$, equivalently the extendability of the $(n;k)$-set $\{\phi_k(P_1), \phi_k(P_2),\cdots, \phi_k(P_n)\}$. We have seen that for large $n\geq q+1$ and for any rational point $P\in E(\f{q})\setminus D$ the vectors $\phi_k(P_1), \phi_k(P_2),\cdots, \phi_k(P_n)$ and $\phi_k(P)$ form an $(n+1; k)$-set. As a consequence, we can re-obtain the deep holes of $C=C_{\Omega}(D, kO)$ in Theorem~\ref{thm:deephole&syn}.

Here raises the interesting problem: Is the $(|E(\f{q})|;k)$-set $\varepsilon=\{\phi_k(P)\mid P\in E(\f{q})\}$ extendable? If not, does there exist an integer $\mathcal{N}_0$ such that any track intersecting with $\varepsilon$ at $\mathcal{N}_0$ or more points must be a part of $\varepsilon$?



\begin{rem}
	The authors~\cite{AGS21,Giu04} studied the extendibility of the $(|E(\f{q})|;k)$-set $\varepsilon=\{\phi_k(P)\mid P\in E(\f{q})\}$.
	The problem for small $k$ is already very difficult, e.g. see~\cite{Giu04} for $k=5$.
	Let $q\geq 121$ be an odd prime power. Let $E$ be an elliptic curve over the finite field $\f{q}$ with non-zero $j$-invariant. Then for $k=3,4,6$ the $(|E(\f{q})|;k)$-set $(\phi_k(P))_{P\in E(\f{q})}$ is non-extendable. And it is conjectured in~\cite{AGS21} that for $k=9$ the $(|E(\f{q})|;9)$-set $(\phi_9(P))_{P\in E(\f{q})}$ is complete.
	
\end{rem}

\begin{cor}\label{thm:completecor}
	Suppose the $(|E(\f{q})|;k)$-set $\varepsilon=\{\phi_k(P)\mid P\in E(\f{q})\}$ is non-extendable and $\mathcal{N}_0$ is the smallest integer such that any track intersecting with $\varepsilon$ at $\mathcal{N}_0$ or more points must be a part of $\varepsilon$. Let $D\subset E(\f{q})\setminus \{O\}$ with cardinality $n=|D|\geq \mathcal{N}_0$. 
	If the residue elliptic curve code $C=C_{\Omega}(D, kO)$ is near-MDS and has covering radius $\rho=k-1$, then the following words 
	\[
	\bigcup_{P\in E(\f{q})\setminus D} C_{\Omega}(D, kO-P)\setminus C
	\]
	form all the deep holes of $C$.
\end{cor}
\begin{proof}
	First, for any $P\in E(\f{q})\setminus D$ and $v\in  C_{\Omega}(D, kO-P)\setminus C$, we have 
	\[
	k-1=\rho\geq d(v, C)\geq \min(d(C), d(C_{\Omega}(D, kO-P)))\geq k-1.
	\]
	So $d(v,C)=k-1$. That is, the word $v$ is a deep hole of $C$. 
	
	Next, we show the completeness. The code $C$ has a parity-check matrix $H=[\phi_k(P)]_{P\in D}$. Let $v$ be a deep hole of $C$. By Proposition~\ref{thm:geometryofdeephole}, $\phi_k(P), P\in D$ and $Hv^T$ form an $n+1$-track in $PG(k-1, \f{q})$. Since the $n+1$-track $\{\phi_k(P)\,|\, P\in D\}\cup\{Hv^T\}$ intersects with $\varepsilon$ at $n\geq \mathcal{N}_0$ points, by the assumption of the corollary, the syndrome $Hv^T$ has to be of the form $\phi_k(Q)$ for some $Q\in E(\f{q})\setminus D$. From the proof of Theorem~\ref{thm:deephole&syn}, $\phi_k(Q)$ is the syndrome of some word $w\in C_{\Omega}(D, kO-Q)\setminus C$. That is,
	 $$Hv^T=\phi_k(Q)=Hw^T$$ 
	 for some $w\in C_{\Omega}(D, kO-Q)\setminus C$. So we have
	 $v\equiv w \bmod C.$ That is, $v\in w+C$ for some $w\in C_{\Omega}(D, kO-Q)\setminus C$. The completeness is proved.
\end{proof}

If such an $\mathcal{N}_0\leq |E(\f{q})|-1$ in the above corollary exists, then Conjecture~\ref{conj:completeness} holds for $C=C_{\Omega}(D, kO)$ where $k\in \{n-3,n-4,n-6\}$. As an extension of~\cite[Theorem~1]{SR86} from the genus $g=0$ to the genus $g=1$, it is interesting to study the non-extendability of $\{\phi_k(P)\mid P\in E(\f{q})\}$ and the existence of $\mathcal{N}_0$ in Corollary~\ref{thm:completecor}.

\section{Conclusion}
In this paper, we studied deep holes of elliptic curve codes. If the covering radius $\rho=k-1$ which typically holds for long residue elliptic $[n,n-k]$-codes, then classes of deep holes and their syndromes are determined. The completeness of the deep holes found was discussed in connection with extendability of $(n;k)$-sets in finite geometry. 
If the deep holes found are not complete, then permutation automorphisms can be applied to obtain more deep holes.

\bibliographystyle{plain}
\bibliography{deephole}
\end{document}